\documentclass[12pt,a4paper]{article}

\usepackage[margin=1in]{geometry}

\usepackage{amsmath,amssymb,amsthm}
\usepackage{mathtools}
\usepackage{mathrsfs}
\usepackage{bm}

\mathtoolsset{showonlyrefs}

\usepackage{graphicx}
\usepackage{dcolumn}

\usepackage[dvipsnames]{xcolor}
\usepackage{hyperref}

\hypersetup{
    colorlinks=true,
    linkcolor=Blue,
    citecolor=Blue,
    urlcolor=Blue
}
\usepackage{authblk}
\usepackage{verbatim}
\usepackage{comment}
\usepackage{siunitx}
\usepackage{enumerate}

\newtheorem{theorem}{Theorem}

\theoremstyle{definition}

\renewcommand{\d}{\mathrm{d}}
\newcommand{\I}{\mathrm{i}}

\newcommand{\ud}{\mathrm{d}}

\newcommand{\ds}{\,\ud s}

\newcommand{\dz}{\,\ud z}
\DeclareMathOperator{\e}{e}

\DeclareMathOperator{\sign}{sign}

\begin{document}

\title{A note on the lower boundary condition for Ekman flows}

\author[1]{\sc Christian Puntini \footnote{Corresponding author: \href{mailto:christian.puntini@univie.ac.at}{christian.puntini@univie.ac.at}}}
\author[2]{\sc Luigi Roberti}

\affil[1]{Faculty of Mathematics, University of Vienna, Oskar--Morgenstern--Platz 1, 1090 Vienna, Austria}
\affil[2]{Institut f\"ur Angewandte Mathematik, Leibniz Universit\"at Hannover, Welfengarten 1, 30167 Hannover, Germany}

\date{}
\maketitle



\noindent\rule{\textwidth}{0.5pt}
\begin{abstract}
We show that the assumption of a stress-free boundary condition at a finite intermediate depth, namely, at the bottom of the Ekman layer, in the analysis of wind-driven ocean flows necessarily leads to an unphysical current profile. Indeed, if the $z$-derivative of the fluid velocity vanishes at a given depth, then this depth necessarily corresponds to a minimum of the velocity profile, with the velocity increasing beneath it. Using a WKB ansatz based on the small variations of the ocean's water density at great depths, we also argue that a no-slip condition at the bottom of the ocean, if sufficiently deep, still effectively implies (up to a very small error) the orthogonality of the Ekman transport and the wind-stress.
\end{abstract}
\noindent\rule{\textwidth}{0.5pt}
\section{Introduction and governing equations}
The study of wind-driven processes at the ocean surface is a classic topic in physical oceanography, with important implications for the ocean circulation and, consequently, Earth’s climate \cite{pedlosky,Siedler2013book,talley}. Its origins can be traced back to the pioneering work of Ekman \cite{Ekman}, who, in 1905, was the first to provide a theoretical description of wind-driven surface currents. His analysis was motivated by the observations made by F. Nansen during the 1893--1896 Arctic expedition aboard the \emph{Fram} vessel, which revealed that sea ice drifts at an angle to the right of the prevailing wind direction. Ekman's explicit solution (see \cite{11} for a concise overview of Ekman's work) applies to the idealised setting of a homogeneous ocean forced by a steady, spatially uniform wind and characterised by a constant vertical eddy viscosity. The interplay between the Coriolis force and the frictional force induced by the wind stress gives rise to the characteristic dynamics of the Ekman currents. The resulting solution exhibits the following key features (see Fig. \ref{ekman spiral}):
\begin{itemize}
\item the surface current is directed at an angle of \(45^\circ\) to the wind, to the right in the Northern Hemisphere and to the left in the Southern Hemisphere;
\item as the depth increases, the current velocity gradually decreases while its direction rotates progressively away from the wind, giving rise to the characteristic \emph{Ekman spiral};
\item the depth-integrated wind-driven transport, known as the \emph{Ekman transport}, is directed at a right angle to the wind, to the right in the Northern Hemisphere and to the left in the Southern Hemisphere.
\end{itemize}
\begin{figure}    \centering    
\includegraphics[width=0.5\linewidth]{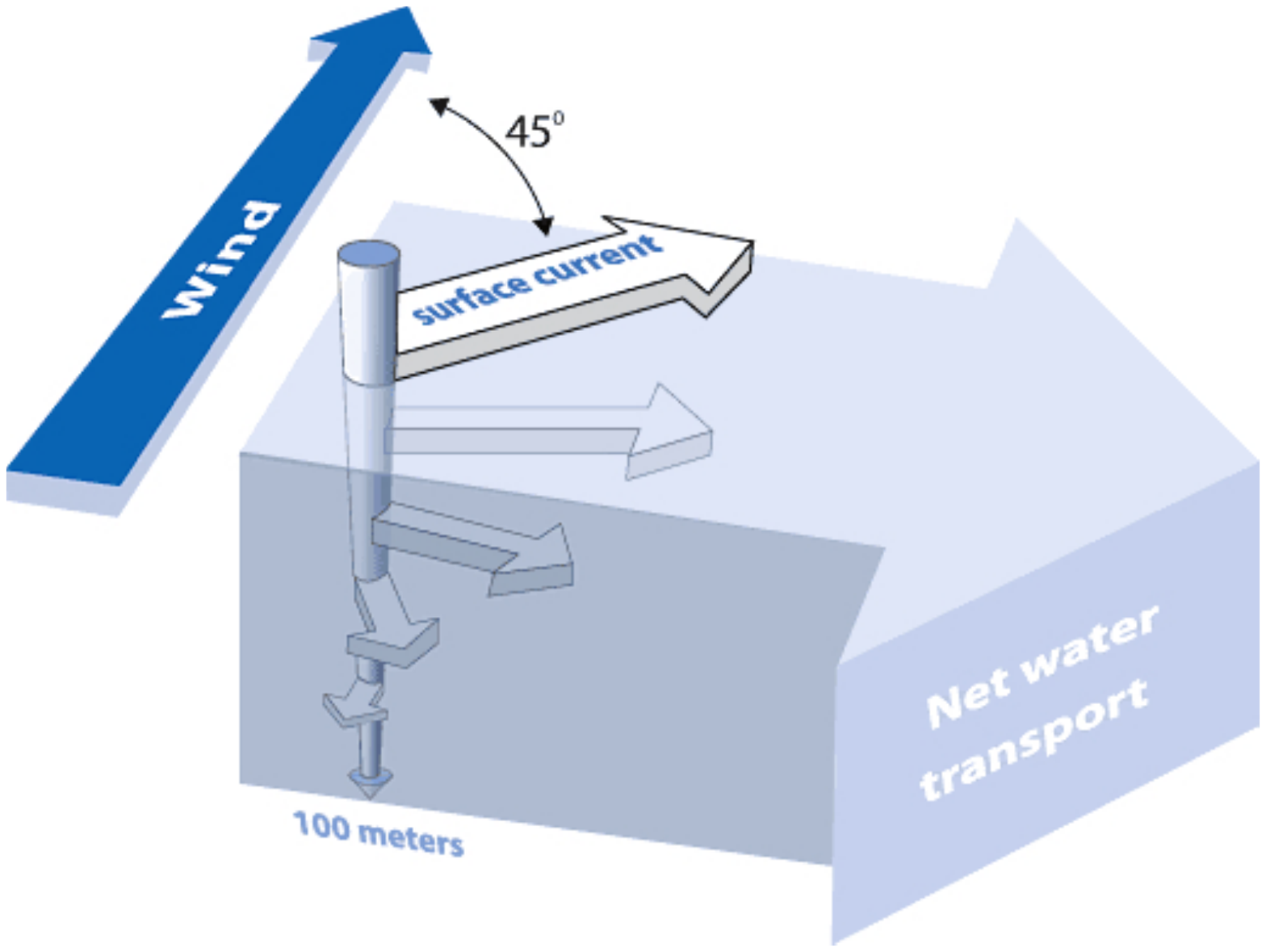}
    \caption{\label{ekman spiral} Schematic representation of the classical Ekman flow: the surface current is deflected by $45^\circ$ relative to the wind direction, and progressively deeper layers exhibit decreasing deflection and velocity. The vertically integrated, wind-driven water transport is directed at right angles to the wind. Credit: NOAA.}
\end{figure}
However, field observations show that the surface deflection angle is far from constant, typically ranging from \(10^\circ\) to \(80^\circ\) (see \cite{NostroJFM,13} and the references therein), in contrast to the fixed \(45^\circ\) predicted by the classical Ekman model. Although it can be proven that the qualitative features of Ekman dynamics persist even when both the water density and the eddy viscosity vary with depth (see \cite{12, NostroJFM}), the observed variability of the deflection angle is widely attributed to the vertical structure of the eddy viscosity. In fact, \cite{NostroJFM} shows that the surface deflection angle is highly sensitive to this structure: different viscosity profiles can produce markedly different departures of the surface current from the wind direction.
Nevertheless, experimental results provide evidence that the Ekman transport is effectively (up to some errors) at a right angle to the blowing wind (see, e.g., \cite{Chereskin, Price1987WindDrivenOC}).

From a theoretical point of view, wind-driven flows are governed by a set of differential equations derived asymptotically from the Navier--Stokes equations, together with a surface boundary condition (coupling the wind and the ocean stresses) and a bottom boundary condition. The Ekman layer is effectively the oceanic boundary layer generated by the transfer of momentum from the wind to the ocean; therefore, it is reasonable to impose the continuity of the stress at the ocean surface. The present letter instead aims to analyse the bottom boundary condition for Ekman flows. Before proceeding further, let us briefly recall the governing equations of Ekman dynamics.

\subsection{The governing equations of Ekman-type flows}
We define a local Cartesian coordinate system $(x, y, z)$ with orthonormal basis vectors $(\mathbf{e}_x, \mathbf{e}_y, \mathbf{e}_z)$. The frame is centred at a point on the spherical surface (excluding the poles, where the horizontal basis vectors are ill-defined) and rotates rigidly with the sphere at the angular velocity $\Omega\approx 7.29\,\cdot\, 10^{-5}\,\rm{s^{-1}}$ about the polar axis (see Fig.~\ref{coordinates}). The velocity components in this frame are denoted by $(u, v, w)$.
\begin{figure}    \centering
    \includegraphics[width=.5\linewidth]{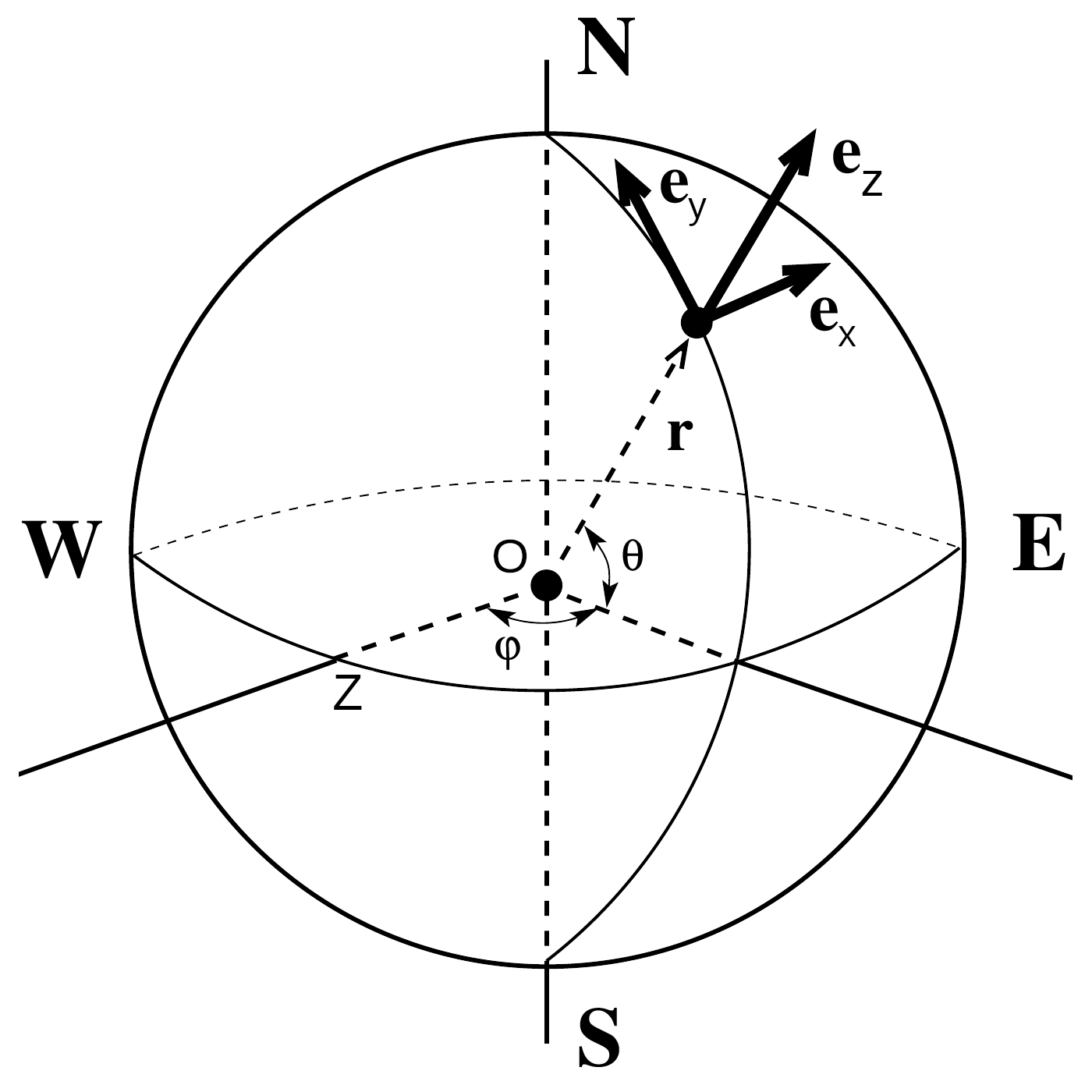}
    \caption{\label{coordinates} The Cartesian coordinate system associated with the rotating Earth. The position vector $r$ specifies the location of a point $P$ with latitude $\theta$ and longitude $\varphi$.  while $Z$ denotes Null Island (coordinates $0^\circ\, \rm{N},0^\circ\, \rm{E} $).}
\end{figure}
The standard approach in the study of Ekman flows is to consider the equations of motion within the so-called $f$-plane approximation. This amounts to approximating a relatively small region of the ocean by a tangent plane and treating the Coriolis parameter $f=2\Omega\sin\theta$ as constant. Furthermore, when analysing wind-driven currents, it is customary to assume that the ocean surface is flat. This assumption is justified by the fact that vertical motions are weaker than horizontal ones and can be rigorously justified through asymptotic arguments; see, e.g., \cite{NARWA}. Under these assumptions, the equations governing the dynamics of Ekman flows reduce to (see \cite{LoperBook})
\begin{equation} \label{Ekman 1}
\left.\begin{aligned}
(m(z)u'(z))' &= -f\rho(z)v(z) \\
(m(z)v'(z))' &= f\rho(z)u(z)
\end{aligned}\right\} \qquad \text{for $-H < z < 0$},
\end{equation}
where a prime denotes differentiation with respect to $z$ (see also \cite{12,13,wang}). In the system \eqref{Ekman 1}, $m$ is the dynamic (vertical) eddy viscosity and $\rho$ is the density of the fluid, both assumed to be depth-dependent. Given $u$ and $v$, and knowing the viscosity profile, the vertical velocity $w$ is obtained via the equation of mass conservation 
\begin{equation}\label{incompressibility}
\frac{\partial u}{\partial x}+\frac{\partial v}{\partial y}+\frac{\partial w}{\partial z}+\frac{\rho'(z)}{\rho(z)}w=0.
\end{equation}
Finally, the pressure is given by
\begin{equation}
 p(z)=p_{\rm atm} + g\int^0_z  \rho(s)\ds.
\end{equation}
where $p_{\rm atm}$ is the atmospheric pressure, assumed to be constant.\\
\noindent
On the flat surface $z=0$, the water's shear stress matches the wind stress $ \bigl(\tau_x,\tau_y\bigr)$:
\begin{equation} \label{surface BC 1}
    m(0)\bigl(u'(0),\,v'(0)\bigr) = \bigl(\tau_x,\tau_y\bigr) .
\end{equation}
Note that, in general, $\bigl(\tau_x,\tau_y\bigr)$ is a function of the wind speed and the surface current $(u(0),v(0))$ and may be given by a (possibly nonlinear) bulk formula, such as in \cite{SL}, but its precise form will not be important for us in this work.

An important \emph{caveat}, which we note here only as a side remark, is that in fact the formulation of the stress in the water, namely, the left-hand side of \eqref{surface BC 1}, is debatable, in the sense that the definition of an eddy viscosity at the surface is not completely formally correct; for a detailed discussion of this fact, we refer the interested reader to \cite{NostroJFM}. Here, we simply fix the boundary condition on the surface in the form \eqref{surface BC 1}, which is the most common one in physical oceanography \cite{12,NARWA,13,talley, vallis}; the precise choice of this condition is, in fact, irrelevant to our subsequent discussion. Indeed, the subject of this paper is not the boundary condition on the surface but rather the one to impose at a lower depth, the choice of which is often motivated by the notion of \emph{Ekman transport}.\\
\noindent
The Ekman transport, in the case of variable density, is defined as
\begin{equation}\label{transport}
\begin{aligned}    
    \mathrm{M_E}&=\int^0_{-H}\rho(z)\bigl(u(z),\,v(z)\bigr)\dz \\
    &=\dfrac{1}{f}\int^0_{-H}\bigl((m(z)v'(z))',\,-(m(z)u'(z))'\bigr)\dz\\
    &=\dfrac{1}{f}\left\{\bigl(\tau_y,\, -\tau_x\bigr)- m(-H)\bigl(v'(-H),\,-u'(-H)\bigr)\right\}.
    \end{aligned}
\end{equation}
As we mentioned earlier, measurements show that the Ekman transport and the wind stress form a $90^\circ$-angle \cite{Chereskin, Price1987WindDrivenOC}. Looking at the formula in \eqref{transport}, it is immediate to observe that the first term $\bigl(\tau_y,\, -\tau_x\bigr)$ is perpendicular to the wind stress vector $\boldsymbol{\tau}=\bigl(\tau_x,\, \tau_y\bigr)$,
and therefore the term $m(-H)\bigl(v'(-H),\,-u'(-H)\bigr)$ should be zero to satisfy the orthogonality condition $\mathrm{M_E}\perp \boldsymbol{\tau}$. This is where the bottom boundary condition enters.

Often, the simplification of infinite depth ($H = \infty$) is made \cite{12,11,NARWA,13,talley}. Since the Ekman flow should vanish at large depths, this translates to the requirement
\begin{equation}
    \lim_{z\to-\infty}(u(z),v(z)) = (0,0);
\end{equation}
then, $\displaystyle \lim_{z\to-\infty}(u'(z),v'(z)) = (0,0)$, as long as the eddy viscosity is ``well behaved" at large depths (which is the case, since the eddy viscosity may be assumed to tend to the constant molecular viscosity at large depths). However, it is of course more physically realistic to consider a finite depth. Then, $z=-H$ can be viewed as the bottom of the ocean, in which case either the no-slip condition
\begin{equation} \label{free slip}
    (u(-H),\,v(-H)) = (0,0)
\end{equation}
or the stress-free condition
\begin{equation}\label{stress free}
    (u'(-H),\,v'(-H)) = (0,0),
\end{equation}
may be imposed; see, e.g., \cite{Deusebio2014Ekman,GerardVaret2006Ekman,Shrira}. The condition \eqref{stress free} ensures the orthogonality condition $\mathrm{M_E}\perp \boldsymbol{\tau}$, but, except in the trivial case of zero wind, implies that the current is not zero at the bottom of the ocean; on the other hand, \eqref{free slip} imposes no motion on the ocean's bed, which is more physically realistic, but has the drawback that, in \eqref{transport}, the Ekman transport is only approximately---though, usually, very nearly---perpendicular to the wind stress (cf. the discussion at the end of the paper). Finally, in several instances where the authors focus on the mixed layer (such as \cite{LewisBelcher,McWilliamsETAL,PriceSundermeyer}), the depth of the Ekman layer is taken \emph{a priori} to be a value $-D\in(-H,0)$, thus an intermediate value between the ocean's bed and the free surface. In this case, a stress-free condition at $-D$ analogous to \eqref{free slip}, namely
\begin{equation}\label{stress free 2}
    (u'(-D),\,v'(-D)) = (0,0),
\end{equation}
is assumed, and in order to match the orthogonality condition $\mathrm{M_E}\perp \boldsymbol{\tau}$, the flow is taken to be negligible for $z<-D$.

The aim of this paper is to show that the boundary condition \eqref{stress free 2}, when naively applied to the theoretical study of Ekman flows, inevitably leads to an inconsistency. In fact, under this assumption, the Ekman flow, which is implicitly assumed to become negligible below the depth $z=-D$, in fact increases in magnitude beneath this depth, effectively giving rise to a “reverse Ekman spiral.”

\section{Main result}

Let $-H < -D < 0$, where $-H$ is the ocean's depth and $-D$ an intermediate depth (usually the bottom of the Ekman layer). It is convenient to reformulate the problem in terms of the complex notation
\begin{equation}
    U = u + \I v, \qquad \bm{\tau} = \tau_x + \I\tau_y;
\end{equation}
with this notation, we can write \eqref{Ekman 1} , \eqref{surface BC 1} and \eqref{stress free 2} as:
\begin{equation} \label{Ekman complex}
\left\{\begin{aligned}
    & (m(z)U'(z))' =  \I f\rho(z)U(z), \quad z\in (-H,0), \\
    & m(0)U'(0) = \bm{\tau}, \\
    & U'(-D) = 0.
\end{aligned}\right.
\end{equation}

\begin{theorem} \label{theorem 1}
Let $-H < -D < 0$ and suppose that $U$ satisfies \eqref{Ekman complex}. Then $|U|$ decreases with depth until $z=-D$, where it reaches its minimum, and increases again with depth below $z=-D$. Similarly, in the northern hemisphere ($f>0$), the angle between $\bm{\tau}$ and $U(z)$ increases with depth until $z=-D$, reaches a maximum there, and subsequently decreases for greater depths.
\end{theorem}
\begin{proof}
We argue along the lines of the proof of Theorem 1 in \cite{NostroJFM}. We multiply the first equation in \eqref{Ekman complex} by $\overline{U(z)}$, integrate from $-D$ to $z\in [-H,0]$, and perform an integration by parts; this yields
\begin{equation} \label{energy}
\begin{aligned}
    m(z)U'(z)\overline{U(z)} &= \int_{-D}^z m(s)|U'(s)|^2\ds + \I f\int_{-D}^z\rho(s)|U(s)|^2\ds.
\end{aligned}
\end{equation}
Writing $U(z)$ in the exponential form
\begin{equation}
    U(z) = r(z)\e^{\I\phi(z)},
\end{equation}
where $r(z) = |U(z)|$ and $\phi(z) = {\rm arg}(U(z))$, plugging into \eqref{energy}, and comparing the real and imaginary parts of the two sides, we obtain
\begin{equation} \label{real part}
    m(z)r'(z)r(z) = \int_{-D}^z m(s)|U'(s)|^2\ds
\end{equation}
and
\begin{equation} \label{imag part}
    m(z)\phi'(z)r(z)^2 = f\int_{-D}^z\rho(s)|U(s)|^2\ds.
\end{equation}
From \eqref{real part}, it follows 
\begin{equation}
    r'(z)\,
    \begin{cases}
        > 0 & \text{if $z\in (-D,0]$}, \\[0.2em]
        = 0 & \text{if $z=-D$}, \\[0.2em]
        < 0 & \text{if $z\in[-H,-D)$},
    \end{cases}
\end{equation}
whereas \eqref{imag part} implies
\begin{equation} \label{rotation}
    \sign(f)\phi'(z)\,
    \begin{cases}
        > 0 & \text{if $z\in (-D,0]$}, \\[0.2em]
        = 0 & \text{if $z=-D$}, \\[0.2em]
        < 0 & \text{if $z\in[-H,-D)$}.
    \end{cases}
\end{equation}
This gives the claim.
\end{proof}
\noindent
See Fig. \ref{spiral_reverse} for a sketch of the behaviour of the solution to \eqref{Ekman complex}.

\begin{figure}    \centering
    \centering
    \includegraphics[width=0.5\linewidth]{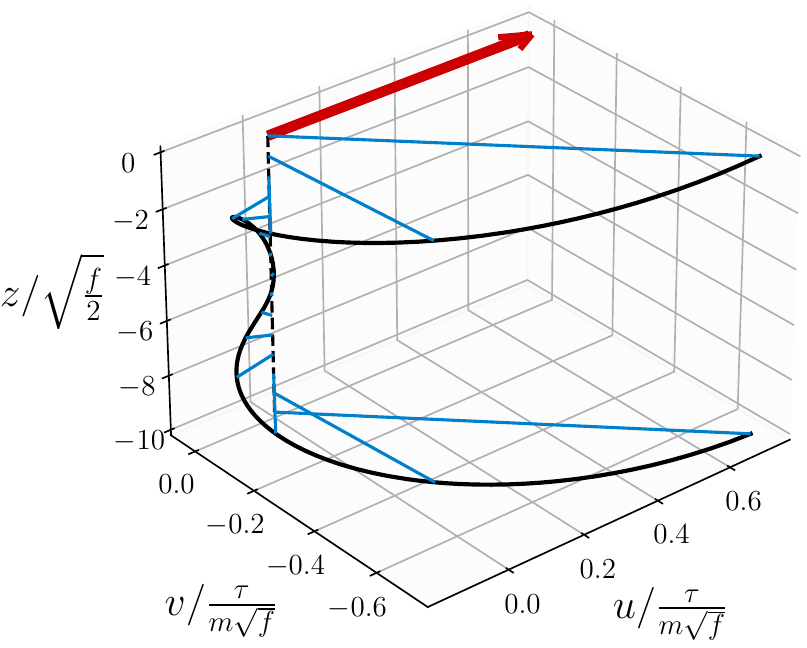}
    \caption{An illustration of the behaviour of the ``reverse Ekman spiral" from Theorem \ref{theorem 1}. Here, for simplicity, $m(z) \equiv m = {\rm const.}$, $f>0$, $\bm{\tau} = (\tau,0)$ for $\tau>0$, and the vertical variable $z$ and the horizontal velocity components are scaled with respect to $\sqrt{f/2}$ and $\tau/(m\sqrt{f})$, respectively. The red arrow denotes the direction of the wind. The derivative of $U = u + \I v$ vanishes for $z = -5\sqrt{f/2}$.}
    \label{spiral_reverse}
\end{figure}

\section{Discussion}\label{conclusion}

In this work, our aim was to highlight the importance of imposing the appropriate boundary condition in theoretical studies, which represent highly idealised real-world scenarios. In particular, we analysed the bottom boundary condition for the Ekman dynamics, and we have proved that assuming \emph{a priori} a certain depth for the Ekman layer, with a stress-free boundary condition, leads to an unrealistic current profile. On the other side, also assuming the depth to be infinite is not physical, and on the mathematical side it implies the cancellation of one of the two linearly independent solutions of the second-order ODE governing the Ekman flows. Although we do not question the use of a stress-free boundary condition in the study or the the parameterisation of the mixed upper surface layer (as, for example, in  \cite{LewisBelcher,McWilliamsETAL,PriceSundermeyer}), we argue that the more formally correct choice is to take a finite-depth ocean with the no-slip boundary condition \eqref{free slip}. In this context, the Ekman depth emerges in the parametrisation of the vertical eddy viscosity $m(z)$. Since the eddy viscosity in this framework could be thought of as the turbulent viscosity generated by the wind stress, the Ekman depth, or the bottom of the Ekman layer, can be identified as the depth below which the wind effects are negligible, and hence such eddy viscosity becomes the very small molecular viscosity. Field measurements confirm this qualitative picture: the eddy viscosity typically increases with depth up to a certain point, beyond which it decreases towards the molecular value \cite{sentchev}. From a theoretical standpoint, this profile is commonly reproduced through the so-called KPP (K-Profile Parametrisation) scheme \cite{Large94}, which prescribes $m(z)$ as a smooth function matching the surface boundary layer physics to the interior, and is widely adopted in ocean models for its ability to capture this non-monotonic behaviour. The resulting shape is illustrated in Fig.~\ref{eddy}.
\begin{figure}    \centering
    \centering
    \includegraphics[width=0.5\linewidth]{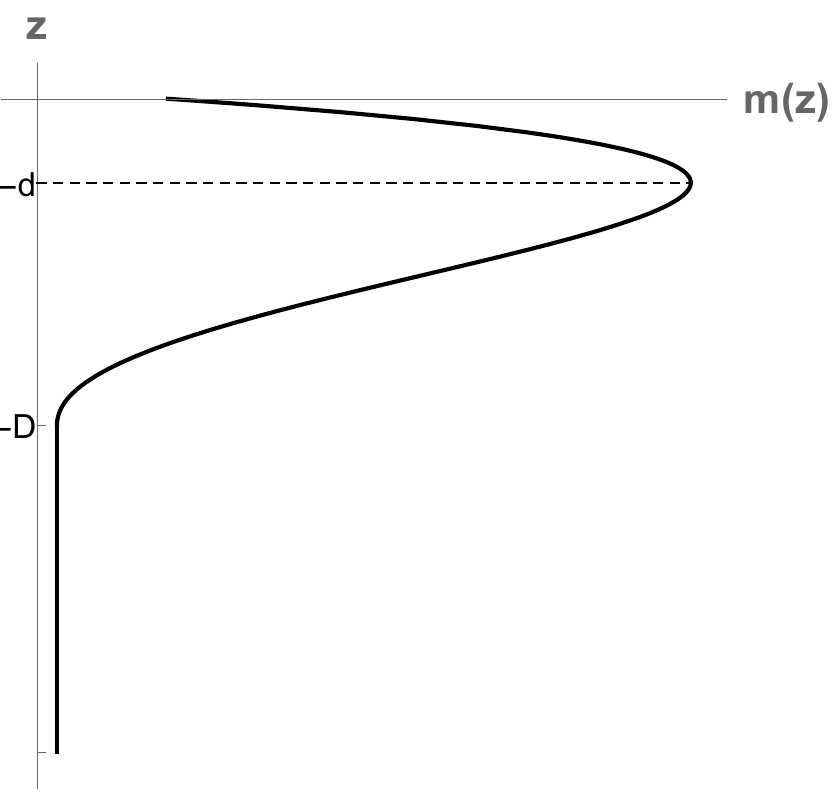}
    \caption{ \label{eddy}Depiction of a typical eddy viscosity profile according to the KPP parametrisation: it increases as depth increases, up to reaching a maximum at depth $-d$, then it decreases to a very small value (that can be thought of as the molecular viscosity) at a certain depth $-D$, below which it attains an almost constant value. As the eddy viscosity is generated by turbulent motion, which in this problem is due by the wind, the depth $-D$ can be considered as the Ekman depth, below which wind effects are negligible. The image is not to scale.}   
\end{figure}
On the other hand, since the Ekman transport is 
\begin{equation}
\begin{aligned}    
    \mathrm{M_E}=\dfrac{1}{f}\left\{\bigl(\tau_y,\, -\tau_x\bigr)- m(-H)\bigl(v'(-H),\,-u'(-H)\bigr)\right\},
    \end{aligned}
\end{equation}
one could argue that assuming a no-slip condition at the ocean's bottom, i.e. $u(-H)=v(-H)=0$, does not guarantee that $\mathrm{M_E}\perp\boldsymbol{\tau}$. In this last part, we show that in the regime $H \gg D$, the derivative $U'(-H)$ is very small, so that the orthogonality condition for the Ekman transport holds up to a small error. The main difficulty stems from the depth-dependence of $\rho$, which does not allow us to find an explicit solution of the governing ODE; however, since density variations are small at great depths (see \cite{talley}), we may use the following WKB ansatz to obtain an estimate on $U'(-H)$.

Let $\mathfrak{m}$ denote the (constant) value of the eddy viscosity below $-D$, and let $U_D := u(-D) + \mathrm{i}\,v(-D)$. The governing ODE \eqref{Ekman 1} then reads
\begin{equation}\label{ekman fondo}
    \mathfrak{m}\,U''(z) = \mathrm{i} f \rho(z)\, U(z), \qquad z \in (-H,-D),
\end{equation}
coupled with the boundary conditions
\begin{equation}
    U(-D) = U_D \quad (\text{given}), \qquad U(-H) = 0.
\end{equation}
 Note that, due to Theorem 1, we have that $|U_D| < |U(0)|$. 
Introducing the local wave-number
\begin{equation}
\kappa(z) = \sqrt{\frac{\mathrm{i} f \rho(z)}{\mathfrak{m}}} = (1+\mathrm{i})\,\mu(z),
\qquad
\mu(z) = \sqrt{\frac{f\rho(z)}{2\mathfrak{m}}} > 0,
\end{equation}
we look for solutions of \eqref{ekman fondo} of WKB form
\begin{equation}
U(z) \approx \kappa(z)^{-\frac{1}{2}} \exp\!\left(\pm \int_{-H}^z \kappa(s)\,\mathrm{d}s\right).
\end{equation}
Substituting this ansatz into \eqref{ekman fondo}, one finds that the terms neglected in the approximation are $O(\kappa'/\kappa^2)$; the WKB approximation is therefore justified provided
\begin{equation}
\left|\frac{\kappa'(z)}{\kappa^2(z)}\right| \ll 1
\qquad \Longleftrightarrow \qquad
\left|\frac{\rho'(z)}{\rho(z)}\right| \ll |\kappa(z)|,
\end{equation}
i.e.\ provided that $\rho$ varies slowly compared to the local wavelength of oscillation---a condition that, as noted above, is satisfied at great depths. With this in hand, we set
\begin{equation}
I(z) = \int_{-H}^{z} \kappa(s)\ds\qquad\text{and}\qquad J = \int_{-H}^{-D} \kappa(s)\ds=I(-D),
\end{equation}
and define the WKB solutions $U_\pm(z) = \kappa(z)^{-\frac{1}{2}}\e^{\,\pm I(z)}$, so that the general solution is $U(z) = a\,U_+(z) + b\,U_-(z)$. Imposing the boundary conditions $U(-H)=0$ (noting that $I(-H)=0$) and $U(-D)=U_D$, we get the following expression for the WKB solution:
\begin{equation}
\begin{aligned}
    U(z) &= U_D \left(\frac{\kappa(-D)}{\kappa(z)}\right)^{\frac{1}{2}}\frac{\sinh\big(I(z)\big)}{\sinh(J)} = U_D \sqrt[4]{\frac{\rho(-D)}{\rho(z)}}\frac{\sinh\big(I(z)\big)}{\sinh(J)},
\end{aligned}
\end{equation}
leading to
\begin{equation}\label{derivative}
\begin{aligned}
    U'(z) &= \frac{U_D\,\kappa(-D)^{\frac{1}{2}}}{\sinh(J)}\bigl(\kappa(z)^{\frac{1}{2}}\cosh\big(I(z)\big) -\tfrac12\,\kappa(z)^{-3/2}\kappa'(z)\,\sinh\big(I(z)\big)\bigr).
\end{aligned}
\end{equation}
Note that the second term, which comes from differentiating the slowly varying WKB prefactor $\kappa(z)^{-\frac{1}{2}}$, has size relative to the other (leading order) term controlled by
\begin{equation}
\begin{aligned}
    \left|\frac{\tfrac12\kappa(z)^{-3/2}\kappa'(z)\sinh(I(z))}{\kappa(z)^{\frac{1}{2}}\cosh(I(z))}\right| &= \left| \frac12\,\frac{\kappa'(z)}{\kappa(z)^{2}}\,\tanh\big(I(z)\big)\right|\leq \left|\frac{\kappa'(z)}{2\,\kappa(z)^{2}}\right|,
\end{aligned}
\end{equation}
which is consistent with the WKB ansatz ($|\rho'(z)/\rho(z)|\ll|\kappa(z)|$). Therefore, dropping this term is consistent with, and not worse than, the WKB approximation already made when constructing $U(z)$ itself. Moreover, note that the second term is exactly zero at the evaluation point $-H$, independently of the size of $\kappa'/\kappa^2$. Indeed, we have that
\(
\tfrac12\,\kappa(-H)^{-3/2}\kappa'(-H)\,\sinh\big(I(-H)\big)=0,
\)
so the formula for $U'(-H)$ below is not merely leading-order WKB, but it is exact at this particular point, given the WKB form of $U(z)$; the WKB error only enters indirectly, through how well $U(z)$ itself approximates the true solution away from $z=-H$. Consequently, for $z=-H$ we have
\begin{equation} \label{bound derivative}
\begin{aligned}
    U'(-H) &\approx\frac{U_D\,\bigl[\kappa(-H)\,\kappa(-D)\bigr]^{\frac{1}{2}}}{\sinh(J)} =U_D\sqrt{\frac{\I f}{\mathfrak{m}}}\dfrac{\sqrt[4]{\rho(-H)\rho(-D)}}{\sinh(J)}.
\end{aligned}
\end{equation}
Note that for constant $\rho$ we have $J= \dfrac{(1+\I) (H-D)}{\delta}$, where $\delta=\sqrt{\dfrac{2\mathfrak{m}}{f\rho}}$, and \eqref{bound derivative} reduces to
\begin{equation}
|U'(-H)| \approx 2\sqrt{2}\,\frac{|U_D|}{\delta}\, e^{\frac{D-H}{\delta}}.\end{equation}
To provide an estimate of $U'(-H)$, we consider the following example for a linearly increasing density: set the bottom of the ocean to be \(
H=1000\,\text{m}\) and the depth below which the wind effects are negligible to be \(D=100\,\text{m}\), and write  \(L\coloneqq H-D=900\,\text{m}\). Moreover, the values for the molecular dynamic viscosity of the water and the mid-latitude Coriolis parameter are, respectively, \(
\mathfrak{m} \approx  10^{-3}\ \mathrm{kg\,m^{-1}\,s^{-1}}\) and \(
f \approx  10^{-4}\ \text{s}^{-1} \). Finally, we assume that the density varies linearly from $\rho(-D)=1020\,\mathrm{kg\,m}^{-3}$ to $\rho(-H)=1050\,\mathrm{kg\,m}^{-3} $. Note that, in general, density variations are much smaller than assumed in this example (see, for example, \cite{talley}). With these data, we have
\[
\mu(z)=\sqrt{\frac{f\,\rho(z)}{2\mathfrak{m}}} \approx 0.23 \sqrt{\rho(z)}.
\]
Since $\rho(z)$ is linear in $z$, we can perform the change of variables $\xi=\rho(z)\in[1020,1050]$, with $\d z = 30\,\d\xi$, so that
\begin{equation}
    J =(1+\I) \int_{-H}^{-D}\mu(z)\dz\approx6.9\,(1+\I)\int_{1020}^{1050}\sqrt{\xi}\,\d \xi \approx 6700\,(1+\I),
\end{equation}
which gives the bound
\[
|U'(-H)| \lesssim 35\;|U_D|\;10^{-2910}\approx 0,\]
ensuring that, if the ocean is sufficiently deep and the Ekman depth is sufficiently smaller than the total depth of the ocean, imposing $U(-H) = 0$ also gives $U'(-H)\approx 0$, and the orthogonality condition between the wind and the Ekman transport is ensured. However, this estimate breaks down in regions where density variations play a more significant role, such as coastal areas, where the water column is also shallower. Indeed, in such regions the angle between the wind and the Ekman transport can deviate from $90^\circ$ (see the discussion in \cite{Lentz2012}), as the wind stress is balanced primarily by bottom stress and along-shelf pressure gradients.

\bigskip

\noindent
\textbf{Data availability statement.} No data were created or analysed in this study.\\
\textbf{Conflicts of Interest/Competing interests.} Both authors declare no conflict of interest.

\bibliographystyle{acm}

\bibliography{apssamp}

\end{document}